\documentclass[aps,prx,reprint,groupedaddress,longbibliography,nofootinbib,floatfix]{revtex4-2}

\usepackage[T1]{fontenc}
\usepackage{lmodern}
\usepackage{amsmath,amssymb,mathtools,bm}
\usepackage{amsthm}
\usepackage{graphicx}
\usepackage{float}
\usepackage{needspace}
\usepackage{xcolor}
\usepackage{microtype}
\usepackage[colorlinks=true,allcolors=blue!55!black]{hyperref}
\hypersetup{
  pdftitle={The Locality Gap: A Thermodynamic Law for Objective Facts},
  pdfauthor={Maxim V. Churilov},
  pdfsubject={Thermodynamics of objective records under restricted control},
  pdfkeywords={quantum Darwinism, thermodynamics of information, objective records, Landauer principle, multipartite correlations}
}
\usepackage[nameinlink,noabbrev]{cleveref}

\graphicspath{{figures/}}

\newcommand{\kB}{k_{\mathrm B}}
\newcommand{\Tr}{\operatorname{Tr}}
\newcommand{\dd}{\mathrm d}
\newcommand{\Pcal}{\mathcal P}
\newcommand{\Tcorr}{\mathcal T}
\newcommand{\Aarch}{\mathcal A}
\newcommand{\Fact}{\mathfrak F}
\newcommand{\Ffree}{\mathcal F}
\newcommand{\bF}{\mathbf F}
\newcommand{\bG}{\mathbf G}

\newcommand{\Iacc}{I_{\mathrm{acc}}}
\newcommand{\pos}[1]{\left[#1\right]_{+}}

\newtheorem{theorem}{Theorem}
\newtheorem{proposition}[theorem]{Proposition}
\newtheorem{corollary}[theorem]{Corollary}

\theoremstyle{definition}

\theoremstyle{remark}

\makeatletter
\AtBeginDocument{%
  \@ifpackageloaded{hyperref}{\hypersetup{hidelinks}}{}%
}
\makeatother

\begin{document}

\title{The Locality Gap: A Thermodynamic Law for Objective Facts}

\author{Maxim V. Churilov}
\email{churilovm1305@gmail.com}
\affiliation{Independent Researcher, Orenburg, Russia}
\date{July 30, 2026}

\begin{abstract}
Objective facts in quantum Darwinism are values recorded redundantly in many independently accessible fragments of the environment.  Redundancy does not create additional logical information, so its thermodynamic meaning has remained unclear.  We show that it creates an exact work asymmetry between controllers with different access architectures.  For a record cloud $\bF$ and a partition $\Pcal$ into blocks that cannot be jointly controlled, the reversible isothermal erasure penalty relative to a global controller is
\begin{equation*}
G_{\Pcal}^{\emptyset}=\kB T\!\left[\sum_{B\in\Pcal}S(F_B)-S(\bF)\right].
\end{equation*}
If the source event $X$ is supplied as catalytic classical side information, the corresponding penalty is $G_{\Pcal}^{X}=\kB T\Tcorr_{\Pcal}(\bF|X)$.  Their difference obeys the central source-loss law
\begin{equation*}
\Fact_{\Pcal}^{T}(X)=G_{\Pcal}^{\emptyset}-G_{\Pcal}^{X}
=\kB T\!\left[\sum_{B\in\Pcal}I(X{:}F_B)-I(X{:}\bF)\right].
\end{equation*}
We call $\Fact_{\Pcal}^{T}$ the thermodynamic factuality charge; equivalently, it is temperature times the event-specific increase in entropy production under partition-local erasure.  It ranges from $-\kB T H(X)$ for perfect secret sharing to $(|\Pcal|-1)\kB T H(X)$ for perfect broadcast objectivity.  Its normalized form defines a thermodynamic record number between $0$ and $|\Pcal|$: perfect secret sharing, unique storage, and perfect broadcasting give $0$, $1$, and $|\Pcal|$, respectively.  We derive Landauer--Darwin bounds from independently decodable records, endpoint-rigidity certificates, exact growth and partition-refinement laws, spectrum-broadcast saturation, and closed formulas for noisy classical and quantum collision models.  For Hamiltonians with cross-block interactions, the entropy-only law acquires one exact interaction-energy correction bounded in magnitude by \(2\lVert H_{\rm int}\rVert\), giving a sharp sign certificate.  The theory does not modify quantum mechanics or posit objective collapse; it identifies the thermodynamic resource generated by redundant records under restricted control.
Pinsker control yields a correlation-adaptive refinement of the
interaction correction and an explicit positivity criterion beyond the
worst-case norm bound.
\end{abstract}

\maketitle

\section{Introduction}

A physical event becomes a public fact only after evidence of it is available in more than one place.  In quantum theory this intuition is made precise by decoherence, quantum Darwinism, and spectrum broadcast structure: a preferred classical variable is imprinted into many fragments of the environment, allowing observers to infer the same value without direct access to the system \cite{Zurek2003,Zurek2009,BlumeKohout2006,BlumeKohout2008,ZwolakQuanZurek2009,ZwolakRiedelZurek2014,Korbicz2014,Brandao2015,LeOlayaCastro2019,Korbicz2021,Touil2022,Touil2024,DoucetDeffner2024}.  Redundancy explains accessibility and consensus.  It does not, however, increase the Shannon information of the underlying event: $R$ copies of one bit still contain one bit of logical uncertainty.  What physical resource, then, is accumulated when one event acquires a large cloud of records?

Thermodynamics suggests that the missing resource is not an additional state variable but an asymmetry between control architectures.  Landauer erasure, reversible computation, and modern quantum thermodynamics relate minimum work to entropy and nonequilibrium free energy \cite{Landauer1961,Bennett1982,Bennett2003,SagawaUeda2009,Parrondo2015,Goold2016,ReebWolf2014,delRio2011,Faist2015,FaistRenner2018}.  Correlations can be converted into work by global control but become inaccessible to modular or local control, producing work deficits and modularity costs \cite{Oppenheim2002,Groisman2005,Horodecki2005,Boyd2018,Guan2025}.  Recent work has also quantified the work cost of quantum measurements, complexity-constrained erasure, and process erasure \cite{LatuneElouard2025,Munson2025,Badhani2026}.  These developments invite a sharper question: \emph{is there a universal thermodynamic signature that distinguishes an objective broadcast record from unique storage or secret sharing?}

We answer this question by comparing four reversible reset tasks.  The same record cloud is erased either globally or block by block, first without and then with the event label supplied as side information.  A single difference of work costs measures all correlations across the controller's partition.  A \emph{double difference} removes correlations unrelated to the event and leaves a signed, target-specific quantity:
\begin{equation}
\boxed{
\Fact_{\Pcal}^{T}(X)
=\kB T\left[\sum_{B\in\Pcal}I(X{:}F_B)-I(X{:}\bF)\right].
}
\label{eq:master}
\end{equation}
Positive values mean that summed blockwise evidence exceeds the information in the union; negative values mean that the union contains an excess unavailable to the separate blocks.  Unique storage gives zero, although zero can also arise from exact cancellation between redundant and synergistic structure.  Perfect broadcast and perfect secret sharing saturate the two universal endpoints.  Thus one thermodynamic coordinate places classical objectivity and privacy-preserving synergy at opposite ends of the same information architecture.

The equality connecting the unconditioned local--global work gap to total correlation is a specialization of correlation free energy and is not claimed here as an isolated novelty.  The contributions of the present work are its source-conditioned subtraction, the operational identification in \cref{eq:master}, its sharp extremal characterization, the Landauer--Darwin objectivity bounds, the dynamic and controller-scale chain rules, and their application to quantum record formation.  To our knowledge, this complete connection between redundant environmental objectivity and a measurable local--global erasure double difference has not previously been formulated.

\begin{figure*}[!t]
\centering
\includegraphics[width=0.98\textwidth]{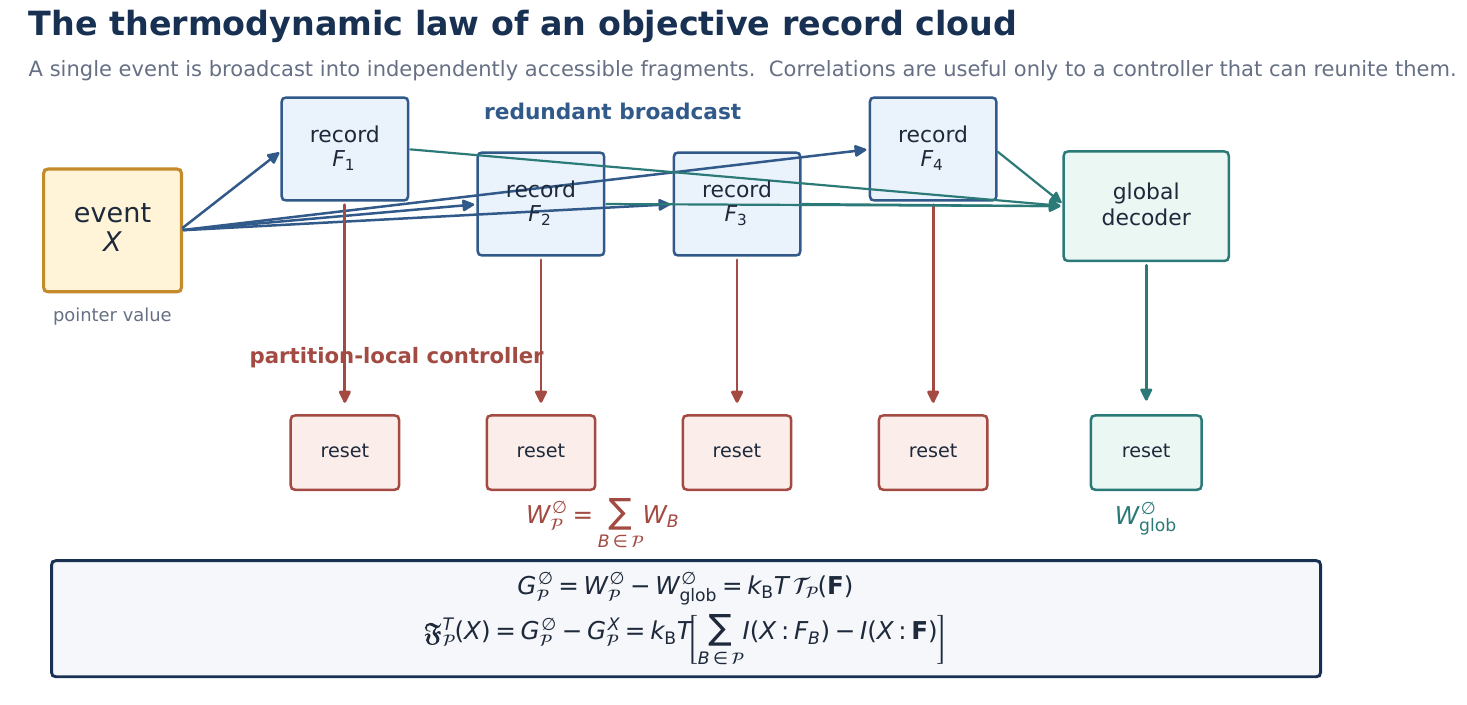}
\caption{Operational setting.  An event $X$ is broadcast into a record cloud $\bF=F_1\cdots F_R$.  A partition-local controller resets each allowed block independently, whereas a global controller first exploits cross-block correlations.  Their work difference is the locality gap $G_{\Pcal}^{\emptyset}$.  Repeating both tasks with $X$ supplied as catalytic classical side information gives $G_{\Pcal}^{X}$; the double difference is the thermodynamic factuality charge \cref{eq:master}.}
\label{fig:concept}
\end{figure*}

\paragraph{Claim map.}
The locality-gap identity, the source-loss and entropy-production forms, the endpoint bounds and rigidity identities, the Landauer--Darwin inequalities, the chain rules, and all displayed finite models are exact within their stated operational assumptions.  Causal sealing is a controller-relative extension of those finite-system quantities.  The factuality-rate criterion for pointer selection is explicitly conjectural and is not used in any proof.

\section{Record clouds and information architecture}
\label{sec:records}

Let $X$ be a finite classical event label with distribution $p_x$, and let $\bF=F_1\cdots F_R$ denote its physical record cloud.  The joint state has classical--quantum form
\begin{equation}
\rho_{X\bF}=\sum_x p_x\,|x\rangle\!\langle x|_X\otimes\rho_{\bF|x}.
\label{eq:cqstate}
\end{equation}
A controller does not necessarily possess arbitrary joint access to all records.  Its access architecture is represented by a partition
\begin{equation}
\Pcal=\{B_1,\ldots,B_m\},\qquad
\bF=F_{B_1}\cdots F_{B_m},
\end{equation}
where a block $F_B$ may be controlled coherently, but different blocks may not be jointly processed in the partition-local protocol.

All logarithms are natural, so entropies are measured in nats.  Define the partition total correlation \cite{Watanabe1960}
\begin{align}
\Tcorr_{\Pcal}(\bF)
&=\sum_{B\in\Pcal}S(F_B)-S(\bF) \notag\\
&=D\!\left(\rho_{\bF}\middle\|\bigotimes_{B\in\Pcal}\rho_{F_B}\right)\ge0,
\label{eq:totalcorr}
\end{align}
and its conditional counterpart
\begin{equation}
\Tcorr_{\Pcal}(\bF|X)
=\sum_xp_x\left[\sum_{B\in\Pcal}S(F_B)_{\rho_x}-S(\bF)_{\rho_x}\right]\ge0.
\label{eq:condtotalcorr}
\end{equation}
The target-specific signed architecture charge is
\begin{equation}
\Aarch_X^{\Pcal}(\bF)
=\sum_{B\in\Pcal}I(X{:}F_B)-I(X{:}\bF).
\label{eq:architecture}
\end{equation}
For $m=2$, this is the interaction information $I(F_{B_1}{:}F_{B_2})-I(F_{B_1}{:}F_{B_2}|X)$.  For general $m$ it is a target-specific multivariate quantity, related in spirit but not identical to a unique partial-information decomposition or to the O-information \cite{Rosas2019}.

\begin{proposition}[Source-loss decomposition]
For every state of the form \cref{eq:cqstate} and every partition $\Pcal$,
\begin{equation}
\boxed{
\Tcorr_{\Pcal}(\bF)
=\Aarch_X^{\Pcal}(\bF)+\Tcorr_{\Pcal}(\bF|X).
}
\label{eq:decomposition}
\end{equation}
Moreover,
\begin{equation}
-H(X)\le\Aarch_X^{\Pcal}(\bF)\le(m-1)H(X).
\label{eq:architecturebounds}
\end{equation}
\end{proposition}

The decomposition separates correlations attributable to the common event $X$ from correlations that remain even after $X$ is specified.  The second term is a nonnegative conditional background.  The first can have either sign: redundant copies make it positive, while secret-sharing structures make it negative.  Proofs of \cref{eq:decomposition,eq:architecturebounds} and all equality conditions are given in \cref{app:information}.

When $H(X)>0$, define the normalized architecture coordinate
\begin{equation}
\phi_X^{\Pcal}=\frac{\Aarch_X^{\Pcal}}{H(X)}\in[-1,m-1].
\label{eq:phi}
\end{equation}
The shifted coordinate
\begin{equation}
\boxed{
\mathcal N_X^{\Pcal}
\equiv 1+\phi_X^{\Pcal}
=1+\frac{\Aarch_X^{\Pcal}}{H(X)}
\in[0,m]
}
\label{eq:recordnumber}
\end{equation}
will be called the \emph{thermodynamic record number}.  It is generally noninteger.  It equals $m$ for perfect broadcasting across $m$ blocks, $1$ for unique storage, and $0$ for perfect secret sharing.  Its thermodynamic meaning is established next.

\section{Thermodynamic reset under restricted control}
\label{sec:thermo}

\subsection{Operational assumptions}

The work identities are exact within a deliberately explicit reset model.  The record Hamiltonian is block additive,
\begin{equation}
H_{\bF}=\sum_{B\in\Pcal}H_{F_B},
\end{equation}
and the target blank is a product $\omega_{\bF}=\bigotimes_{B\in\Pcal}\omega_{F_B}$.  The nonequilibrium free energy at bath temperature $T$ is
\begin{equation}
\Ffree_T(\rho)=\Tr(H\rho)-\kB T S(\rho).
\label{eq:freeenergy}
\end{equation}
In the reversible isothermal limit, the minimum average work needed to transform $\rho$ into $\omega$ is $\Ffree_T(\omega)-\Ffree_T(\rho)$.  This is the standard asymptotic or quasistatic regime; one-shot, finite-bath, finite-time, and complexity restrictions generally add corrections \cite{ReebWolf2014,Faist2015,FaistRenner2018,Munson2025,Badhani2026}.

A \emph{partition-local reset} is a tensor product of block protocols with no cross-block communication, coherent gate, shared memory, or feed-forward that could exploit correlations between distinct blocks.  Thermal ancillas, work stores, and controller memories begin product across blocks and may not be retained as a hidden shared resource.  A \emph{global reset} may reversibly compress the entire cloud before erasure.  These are two physically different controller architectures, not two descriptions of the same unrestricted operation.

\subsection{The locality gap}

The minimum partition-local and global reset works are
\begin{align}
W_{\Pcal}^{\emptyset}
&=\sum_{B\in\Pcal}\left[\Ffree_T(\omega_{F_B})-\Ffree_T(\rho_{F_B})\right],
\label{eq:wlocal}\\
W_{\mathrm{glob}}^{\emptyset}
&=\Ffree_T(\omega_{\bF})-\Ffree_T(\rho_{\bF}).
\label{eq:wglobal}
\end{align}
The superscript $\emptyset$ denotes absence of side information about the source event.

\begin{theorem}[Locality-gap law]
Under the assumptions above,
\begin{equation}
\boxed{
G_{\Pcal}^{\emptyset}
\equiv W_{\Pcal}^{\emptyset}-W_{\mathrm{glob}}^{\emptyset}
=\kB T\,\Tcorr_{\Pcal}(\bF).
}
\label{eq:localitygap}
\end{equation}
\end{theorem}

The energy terms cancel because both the Hamiltonian and blank are additive.  The remaining excess work is exactly the correlation free energy inaccessible to the partition-local controller.  Equation \eqref{eq:localitygap} is closely related to thermodynamic work deficit and modularity dissipation \cite{Oppenheim2002,Groisman2005,Horodecki2005,Boyd2018,Guan2025}.  Its role here is to provide the thermodynamic backbone for a target-specific subtraction.

\subsection{Source-assisted reset and the law of factuality}

Suppose the event label $X$ remains available as classical side information.  Operationally, the reset protocol is selected conditionally on $x$ while the source register is returned unchanged.  Both assisted baselines receive the same set of classical control copies and return them unchanged.  Communication, fanout, or storage used to supply those copies is therefore held fixed between the local and global comparisons and is not included in the work difference.  The conditional minimum works are
\begin{align}
W_{\Pcal}^{X}
&=\sum_xp_x\sum_{B\in\Pcal}
\left[\Ffree_T(\omega_{F_B})-\Ffree_T(\rho_{F_B|x})\right],\\
W_{\mathrm{glob}}^{X}
&=\sum_xp_x
\left[\Ffree_T(\omega_{\bF})-\Ffree_T(\rho_{\bF|x})\right].
\end{align}
Hence
\begin{equation}
G_{\Pcal}^{X}
\equiv W_{\Pcal}^{X}-W_{\mathrm{glob}}^{X}
=\kB T\,\Tcorr_{\Pcal}(\bF|X).
\label{eq:conditionalgap}
\end{equation}

\Needspace{8\baselineskip}
\begin{theorem}[Source-loss law of factuality]
The change in the locality gap when the source event ceases to be available is
\begin{align}
\Fact_{\Pcal}^{T}(X{:}\bF)
&\equiv G_{\Pcal}^{\emptyset}-G_{\Pcal}^{X} \notag\\
&=\boxed{\kB T\left[\sum_{B\in\Pcal}I(X{:}F_B)-I(X{:}\bF)\right]}.
\label{eq:factuality}
\end{align}
It obeys the sharp bounds
\begin{equation}
\boxed{
-\kB T H(X)
\le\Fact_{\Pcal}^{T}(X{:}\bF)
\le(m-1)\kB T H(X).
}
\label{eq:bounds}
\end{equation}
\end{theorem}

\begin{corollary}[Entropy-production form]
Using the globally reversible reset as the zero-production baseline, define $\Tcorr_{\Pcal}(\bF|\emptyset)\equiv\Tcorr_{\Pcal}(\bF)$.  For $Y\in\{\emptyset,X\}$,
\begin{equation}
\Sigma_{\Pcal}^{Y}=\frac{G_{\Pcal}^{Y}}{T}
=\kB\Tcorr_{\Pcal}(\bF|Y),
\end{equation}
so
\begin{equation}
\boxed{\Fact_{\Pcal}^{T}
=T\left(\Sigma_{\Pcal}^{\emptyset}-\Sigma_{\Pcal}^{X}\right).}
\label{eq:entropyproduction}
\end{equation}
Thus factuality is temperature times the event-specific increase in unavoidable entropy production caused by source loss.
\end{corollary}

We call $\Fact_{\Pcal}^{T}$ the \emph{thermodynamic factuality charge}.  It is a signed work double difference, not stored internal energy.  At the upper endpoint, every block individually contains a complete record of $X$; at the lower endpoint, no block contains any information about $X$, while the union contains all of it.  For $m\ge2$ and a cq state with all $p_x>0$, upper saturation is equivalent to pairwise orthogonality of the conditional supports on every block.  Lower saturation is equivalent to $\rho_{F_B|x}$ being independent of $x$ for every block and the global conditional states being perfectly distinguishable.  Thus the two endpoints are exact thermodynamic signatures of perfect broadcast and perfect secret sharing, respectively.

The adjective ``factual'' is reserved for the positive, redundancy-dominated regime, while operational objectivity additionally requires independently accessible decoders as imposed below.  Negative values are not violations of thermodynamics; they say that revealing $X$ \emph{creates} conditional cross-block correlations that were absent without the key, precisely as in synergistic encodings.

\begin{theorem}[Endpoint rigidity]
Let $m\ge2$ and define the upper and lower thermodynamic deficits
\begin{equation}
\Delta_+=(m-1)\kB T H(X)-\Fact_{\Pcal}^{T},
\qquad
\Delta_-=\Fact_{\Pcal}^{T}+\kB T H(X).
\end{equation}
For every anchor block $B_0\in\Pcal$,
\begin{align}
\frac{\Delta_+}{\kB T}
&=\sum_{B\ne B_0}\left[H(X)-I(X{:}F_B)\right]
+I(X{:}\bF)-I(X{:}F_{B_0}),
\label{eq:upperrigidity}\\
\frac{\Delta_-}{\kB T}
&=\sum_{B\in\Pcal}I(X{:}F_B)
+H(X)-I(X{:}\bF).
\label{eq:lowerrigidity}
\end{align}
Every term on the right is nonnegative.  Consequently, $\Delta_+\le\kB T\varepsilon$ certifies $I(X{:}F_B)\ge H(X)-\varepsilon$ and $I(X{:}\bF)-I(X{:}F_B)\le\varepsilon$ for every block.  Likewise, $\Delta_-\le\kB T\varepsilon$ certifies $I(X{:}F_B)\le\varepsilon$ for every block and $I(X{:}\bF)\ge H(X)-\varepsilon$.
\end{theorem}

The endpoints are therefore rigid rather than merely extremal: a charge close to the broadcast ceiling forces every block to be an almost complete record, while a charge close to the secret-sharing floor forces every block to be almost ignorant although the union is almost complete.  The identities are proved in \cref{app:information}.

\section{Landauer--Darwin bounds and exact broadcast saturation}
\label{sec:bounds}

Quantum Darwinism quantifies redundancy by counting disjoint fragments from which the pointer value can be inferred with a prescribed information deficit \cite{Zurek2009,ZwolakRiedelZurek2014,Korbicz2021,Touil2022,Chisholm2024}.  We call $X$ an $(R,\delta)$-public record relative to the controller architecture if there exist $R$ disjoint blocks $B_1,\ldots,B_R$ such that
\begin{equation}
\Iacc(X{:}F_{B_j})\ge(1-\delta)H(X),
\qquad j=1,\ldots,R.
\label{eq:redundancycondition}
\end{equation}
Because $I(X{:}F_{B_j})\ge\Iacc(X{:}F_{B_j})$ and $I(X{:}\bF)\le H(X)$, \cref{eq:factuality} gives
\begin{equation}
\boxed{
\frac{\Fact_{\Pcal}^{T}}{\kB T}
\ge\left[R(1-\delta)-1\right]H(X).
}
\label{eq:landauerdarwin}
\end{equation}
This is the \emph{Landauer--Darwin bound}.  In terms of \cref{eq:recordnumber}, the same result reads $\mathcal N_X^{\Pcal}\ge R(1-\delta)$: $R$ independently decodable fragments force at least $R(1-\delta)$ thermodynamic records, without requiring perfect copies.  For the unconditioned locality gap, nonnegativity of total correlation yields
\begin{equation}
\frac{G_{\Pcal}^{\emptyset}}{\kB T}
\ge\pos{R(1-\delta)-1}H(X).
\label{eq:rawbound}
\end{equation}
The positive part is required in \cref{eq:rawbound} but not in the signed factuality bound \cref{eq:landauerdarwin}.

A directly operational variant follows from Fano's inequality.  If $X$ has $d_X$ values and each block admits a decoder with error probability $p_{e,B}$, then \cite{CoverThomas2006}
\begin{equation}
I(X{:}F_B)\ge H(X)-h_2(p_{e,B})-p_{e,B}\ln(d_X-1),
\end{equation}
so
\begin{align}
\frac{\Fact_{\Pcal}^{T}}{\kB T}
&\ge(m-1)H(X)\notag\\
&\quad-\sum_{B\in\Pcal}
\left[h_2(p_{e,B})+p_{e,B}\ln(d_X-1)\right].
\label{eq:fano}
\end{align}
For a binary event and a common error bound $\epsilon$, this becomes $(m-1)H(X)-m h_2(\epsilon)$.

The ideal limit is a spectrum broadcast structure
\begin{equation}
\rho_{X\bF}=\sum_xp_x|x\rangle\!\langle x|
\otimes\bigotimes_{j=1}^{R}\rho_{F_j|x},
\quad
\rho_{F_j|x}\rho_{F_j|x'}=0\;(x\ne x').
\label{eq:sbs}
\end{equation}
Conditional independence makes $\Tcorr(\bF|X)=0$, while orthogonality gives $I(X{:}F_j)=I(X{:}\bF)=H(X)$.  Consequently,
\begin{equation}
\boxed{
G_{\Pcal}^{\emptyset}=\Fact_{\Pcal}^{T}
=(R-1)\kB T H(X).
}
\label{eq:sbssaturation}
\end{equation}
For one unbiased bit, every independently controlled record after the first contributes exactly one Landauer unit, $\kB T\ln2$, to the local--global work asymmetry.

Importantly, maximal factuality does not require conditional independence: it requires complete local records.  Conditional correlations contribute additionally to $G_{\Pcal}^{\emptyset}$ but cancel from $\Fact_{\Pcal}^{T}$.  The source-loss subtraction therefore isolates the architecture of evidence about $X$ from unrelated multipartite structure.

\section{Growth laws and thermodynamic relativity of control}
\label{sec:dynamics}

\subsection{Adding one record}

For singleton blocks, write $\bF_n=F_1\cdots F_n$.  Total correlation obeys the exact chain rule
\begin{equation}
\Tcorr(\bF_{n+1})-\Tcorr(\bF_n)
=I(F_{n+1}{:}\bF_n),
\label{eq:tcincrement}
\end{equation}
which translates into an incremental locality cost
\begin{equation}
G_{n+1}^{\emptyset}-G_n^{\emptyset}
=\kB T I(F_{n+1}{:}\bF_n).
\label{eq:workincrement}
\end{equation}
The factuality increment is
\begin{align}
\frac{\Fact_{n+1}^{T}-\Fact_n^{T}}{\kB T}
&=I(X{:}F_{n+1})-I(X{:}F_{n+1}|\bF_n)
\label{eq:factincrementa}\\
&=I(F_{n+1}{:}\bF_n)-I(F_{n+1}{:}\bF_n|X).
\label{eq:factincrementb}
\end{align}
The sign is not fixed for arbitrary states.  In a branching or conditional-product record process, however,
$I(F_{n+1}{:}\bF_n|X)=0$, and each added fragment increases factuality by the mutual information it shares with the existing record cloud.  This gives a precise monotonicity law for idealized broadcast formation without assuming that each record is perfect.

\subsection{Changing controller resolution}

Let a fine partition be coarsened by merging two blocks $B$ and $C$.  The same physical state then satisfies
\begin{align}
G_{\mathrm{fine}}^{\emptyset}-G_{\mathrm{coarse}}^{\emptyset}
&=\kB T I(F_B{:}F_C),
\label{eq:partitionraw}\\
G_{\mathrm{fine}}^{X}-G_{\mathrm{coarse}}^{X}
&=\kB T I(F_B{:}F_C|X),
\label{eq:partitioncond}\\
\Fact_{\mathrm{fine}}^{T}-\Fact_{\mathrm{coarse}}^{T}
&=\kB T\left[I(F_B{:}F_C)-I(F_B{:}F_C|X)\right].
\label{eq:partitionfact}
\end{align}
Thus thermodynamic factuality is relative to the set of records a controller can jointly manipulate.  This is not observer dependence in the epistemic sense: the state and all entropies are fixed, while the allowed operation algebra changes.

For $R$ perfect copies grouped into $m$ nonempty controllable blocks,
\begin{equation}
\Fact_{\Pcal}^{T}=(m-1)\kB T H(X),
\label{eq:blockbroadcast}
\end{equation}
regardless of the number of copies inside each block.  One global block has zero locality gap; $R$ isolated blocks have the maximal $(R-1)\kB T H(X)$ gap.

\begin{figure}[!tbp]
\centering
\includegraphics[width=\columnwidth]{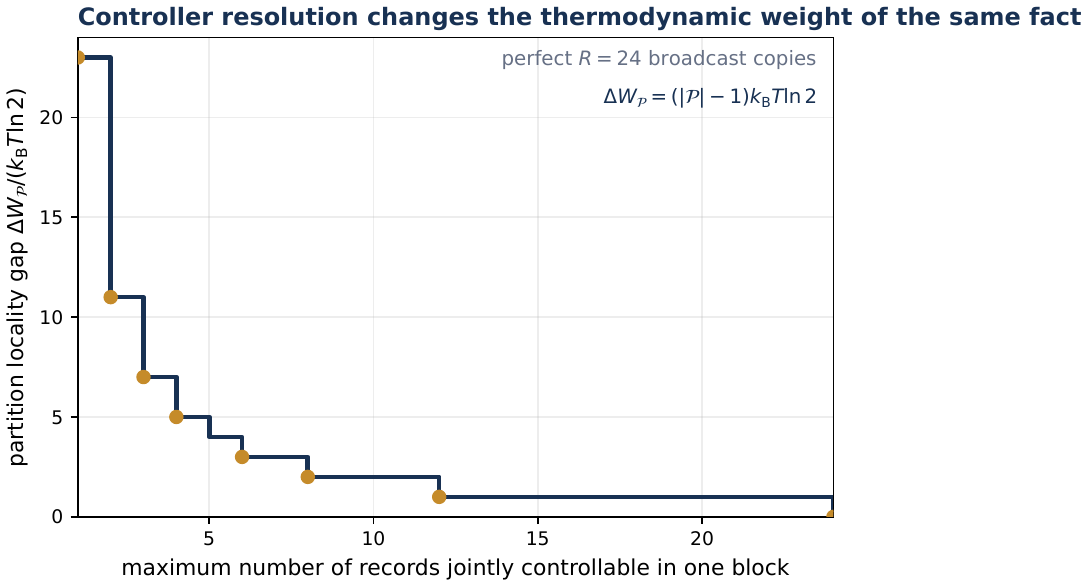}
\caption{Controller-scale flow for $24$ perfect records of one bit.  If at most $b$ records can be jointly processed, the minimum number of blocks is $\lceil24/b\rceil$, and the locality gap is $(\lceil24/b\rceil-1)\kB T\ln2$.}
\label{fig:partition}
\end{figure}

\section{Broadcast--secret-sharing duality}
\label{sec:duality}

\Cref{eq:bounds} places three qualitatively different information architectures on one exact thermodynamic axis.

\emph{Perfect broadcast:} every block determines $X$, so $\phi_X=m-1$.  The event is locally public and maximally expensive to erase without reuniting the records.

\emph{Unique storage:} exactly one block carries $X$ and the others carry no correlated evidence, so $\phi_X=0$.  Information exists, but there is no redundant fact.  The converse is not general: $\phi_X=0$ may also reflect cancellation between redundant and synergistic contributions.

\emph{Perfect secret sharing:} no block reveals $X$, while the union determines it, so $\phi_X=-1$.  Here loss of the source key reduces rather than increases the locality gap.

\begin{figure}[H]
\centering
\includegraphics[width=\columnwidth]{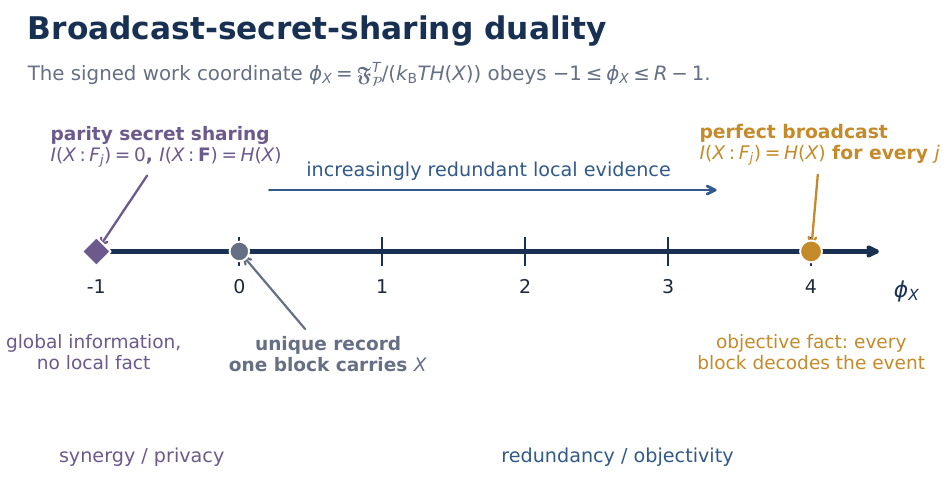}
\caption{The normalized factuality charge distinguishes redundant evidence from synergistic encoding.  For $R=5$ blocks, perfect broadcast saturates $\phi_X=4$, unique storage gives $0$, and parity secret sharing saturates $-1$.}
\label{fig:duality}
\end{figure}

This duality connects classical objectivity to the no-broadcasting theorem.
\begin{proposition}[Classicality of distributed assistance]
Suppose an unknown source memory is to be supplied nondestructively to two or more independent block controllers, with each controller receiving a complete local copy of the source state for every member of an ensemble $\{\sigma_x\}$.  Such exact catalytic distribution is possible only if the states $\sigma_x$ commute.
\end{proposition}
This is the no-broadcasting theorem \cite{Barnum1996} applied to the side-information channel.  It shows that the classical register $X$ in the four-work protocol is not merely a simplifying notation: exact independent assistance cannot be distributed in the same way for an arbitrary unknown quantum source.  Independently, maximal positive factuality requires each record block to contain perfectly distinguishable conditional states.  A channel that produced that endpoint for arbitrary noncommuting inputs would again broadcast them.  The saturating variable must therefore be classical---equivalently, associated with a commuting pointer algebra.  In this precise sense, both the operational baseline and the positive endpoint of the thermodynamic axis select classical information, whereas the negative endpoint is the architecture exploited by privacy and error-correcting codes.

The factuality charge is invariant under local isometries and additive for independent record clouds with product access partitions.  It is not, in general, monotone under arbitrary local noise: local processing may destroy redundant evidence or, by removing synergistic details, change the signed balance.  We therefore do not identify it with a universal resource monotone.  It is instead the exact quantity measured by a specified thermodynamic comparison.

\section{Exactly solvable record models}
\label{sec:models}

\subsection{Perfect classical broadcast}

For $R$ copies of an arbitrary classical variable,
\begin{equation}
P(f_1,\ldots,f_R)=\sum_xp_x\prod_{j=1}^{R}\delta_{f_j,x}.
\end{equation}
Each marginal has entropy $H(X)$ and the joint cloud also has entropy $H(X)$.  Therefore
\begin{equation}
W_{\mathrm{loc}}^{\emptyset}-W_{\mathrm{glob}}^{\emptyset}
=(R-1)\kB T H(X).
\end{equation}
For an unbiased bit, a global controller reversibly maps $|x\rangle^{\otimes R}$ to $|x\rangle|0\rangle^{\otimes R-1}$ and erases one bit; isolated controllers erase $R$ bits.

\subsection{Independent binary-symmetric records}

Let $X$ be an unbiased bit and
\begin{equation}
F_j=X\oplus N_j,
\qquad N_j\stackrel{\mathrm{iid}}{\sim}\mathrm{Bernoulli}(\epsilon).
\end{equation}
The records are conditionally independent, so $G^{X}=0$ and $G^{\emptyset}=\Fact^T$.  Since
\begin{equation}
I(X{:}F_j)=\ln2-h_2(\epsilon),
\end{equation}
the exact result is
\begin{equation}
\boxed{
\frac{\Fact_R^T}{\kB T}
=R\left[\ln2-h_2(\epsilon)\right]-\ln2+H(X|F_{1:R}).
}
\label{eq:bsc}
\end{equation}
The remaining conditional entropy is a one-dimensional binomial sum derived in \cref{app:models}.  As $R$ grows, the joint decoder identifies $X$ with vanishing error for every $\epsilon<1/2$, while each additional noisy record contributes asymptotically $\ln2-h_2(\epsilon)$ nats of factuality.

\begin{figure}[H]
\centering
\includegraphics[width=\columnwidth]{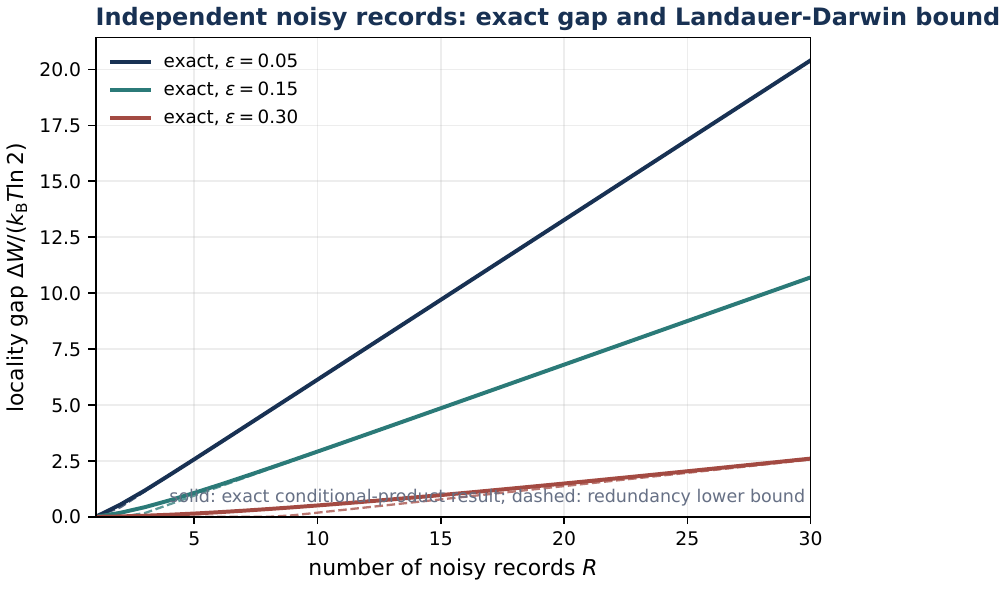}
\caption{Exact factuality/locality gap for independent binary-symmetric records.  Solid curves show \cref{eq:bsc}; dashed curves show the lower bound obtained by replacing $I(X{:}F_{1:R})$ with $H(X)=\ln2$.}
\label{fig:noisy}
\end{figure}

\subsection{Pure-state quantum collision model}

Consider equal-prior branches whose $j$th environmental record is $|e_0\rangle$ or $|e_1\rangle$, with real overlap
\begin{equation}
c=|\langle e_0|e_1\rangle|\in[0,1].
\end{equation}
The conditional cloud states are products $|e_x\rangle^{\otimes R}$, so $\Tcorr(\bF|X)=0$.  A mixture of two equiprobable pure states with overlap $q$ has entropy $h_2[(1+q)/2]$.  Hence
\begin{equation}
\boxed{
\frac{G_R^{\emptyset}}{\kB T}
=\frac{\Fact_R^T}{\kB T}
=R h_2\!\left(\frac{1+c}{2}\right)
-h_2\!\left(\frac{1+c^R}{2}\right).
}
\label{eq:collision}
\end{equation}
The orthogonal limit $c=0$ recovers $(R-1)\ln2$.  For weak individual records $c\lesssim1$, a macroscopic work asymmetry still develops through repeated broadcasting.

\begin{figure}[H]
\centering
\includegraphics[width=\columnwidth]{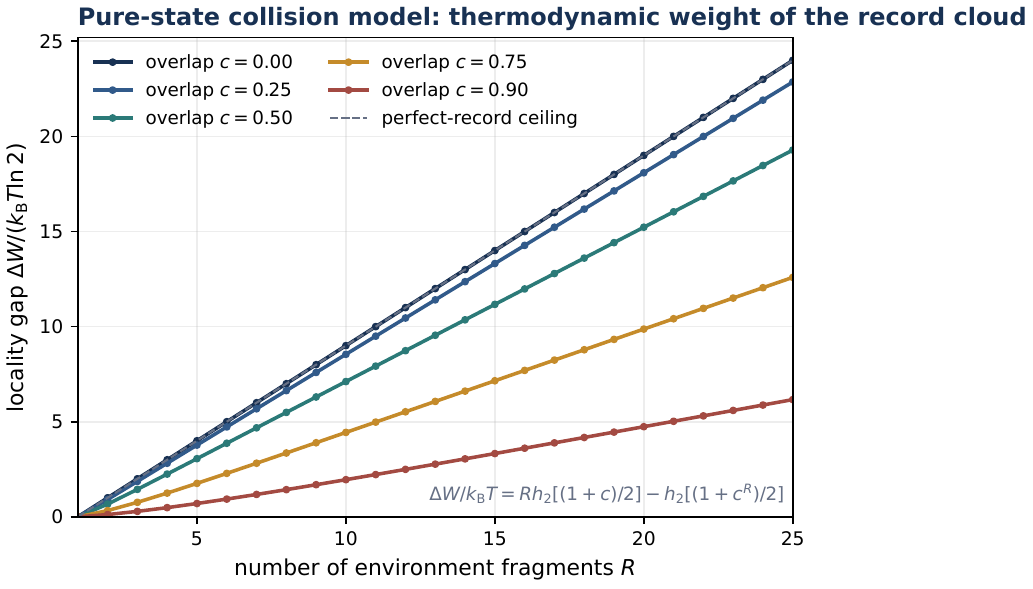}
\caption{Exact source-loss charge for the pure-state collision model.  The dashed ceiling is perfect one-bit broadcasting, $(R-1)\kB T\ln2$.}
\label{fig:collision}
\end{figure}

\subsection{Parity secret sharing}

Let $F_1,\ldots,F_R$ be independent fair bits and define
\begin{equation}
X=F_1\oplus\cdots\oplus F_R.
\end{equation}
Unconditionally, $\Tcorr(\bF)=0$.  Conditioned on $X$, the parity constraint produces $\Tcorr(\bF|X)=\ln2$.  Thus
\begin{equation}
\Aarch_X(\bF)=-\ln2,
\qquad
\Fact^T=-\kB T\ln2,
\end{equation}
which saturates the universal lower bound.  The same event information that appears as an objective fact under broadcasting appears as a private global constraint under parity encoding.

All exact identities were independently stress-tested on $20{,}000$ random classical distributions and $2{,}000$ random noncommuting cq ensembles.  Maximum residuals were below $1.1\times10^{-14}$; methods, fixed seed, diagnostic plots, and full results are given in \cref{app:numerics} and the reproducibility archive.

\section{Experimental protocol}
\label{sec:experiment}

The central quantity is a four-work double difference and can be measured without assigning an ontology to ``information.''  Prepare the same cq record cloud repeatedly and estimate:
\begin{equation}
W_{\Pcal}^{\emptyset},\quad W_{\mathrm{glob}}^{\emptyset},
\quad W_{\Pcal}^{X},\quad W_{\mathrm{glob}}^{X}.
\end{equation}
The first pair resets the cloud without the event label.  The second pair uses a classical control line carrying $x$, returned to its initial state.  Then
\begin{equation}
\Fact_{\Pcal}^{T}
=(W_{\Pcal}^{\emptyset}-W_{\mathrm{glob}}^{\emptyset})
-(W_{\Pcal}^{X}-W_{\mathrm{glob}}^{X}).
\label{eq:fourwork}
\end{equation}
A reversible global protocol first compresses common information; a partition-local protocol deliberately forbids cross-block gates.  Work can be inferred from a calibrated work storage system, quasistatic free-energy integration, or fluctuation-relation estimators.  Finite-time dissipation produces an excess above the reversible values and should be removed by protocol-time extrapolation or separately modeled \cite{Berut2012,Hong2016,Chiribella2022,LatuneElouard2025}.

A minimal demonstration uses $R$ qubit memories.  A controlled fanout prepares the classical mixture of $|0\rangle^{\otimes R}$ and $|1\rangle^{\otimes R}$.  Independent reset costs $R\kB T\ln2$ for degenerate memories, whereas a global inverse fanout followed by one reset costs $\kB T\ln2$.  Supplying $X$ allows every record to be reversibly uncomputed, giving $G^X=0$.  The predicted double difference is therefore $(R-1)\kB T\ln2$ with no fitted constant.

Superconducting circuits offer programmable fanout, controllable global versus local gates, and cryogenic reset; trapped ions and photonic registers offer high-fidelity record preparation; spin environments around nitrogen-vacancy centers already provide experimental access to quantum-Darwinism observables \cite{Unden2019}.  A decisive test should vary $R$, record noise, and the block partition independently.  The theory predicts the full staircase in \cref{fig:partition}, the noisy curves in \cref{fig:noisy}, and the quantum-overlap law in \cref{eq:collision}.

\section{Causal sealing and residual coherence}
\label{sec:causal}

The work law does not assert that unitary quantum information is destroyed.  In a closed system, branching may merely move phase information into inaccessible correlations.  Consider
\begin{equation}
|\Psi\rangle
=\frac{1}{\sqrt2}\left(
|0\rangle\bigotimes_{j=1}^{R}|e_j^{0}\rangle
+|1\rangle\bigotimes_{j=1}^{R}|e_j^{1}\rangle
\right).
\label{eq:branchstate}
\end{equation}
If a set $U$ of fragments is inaccessible and traced out, the off-diagonal branch operator on the accessible system is multiplied exactly by
\begin{align}
C_U&=\prod_{j\in U}\langle e_j^{1}|e_j^{0}\rangle,\notag\\
|C_U|&=\exp\!\left[-\frac12\sum_{j\in U}\xi_j\right],
\qquad
\xi_j=-2\ln|\langle e_j^{1}|e_j^{0}\rangle|.
\label{eq:coherencefactor}
\end{align}
This is a statement about the reduced state after passive inaccessibility, not a universal upper bound on every assisted recovery protocol.

\begin{figure}[H]
\centering
\includegraphics[width=\columnwidth]{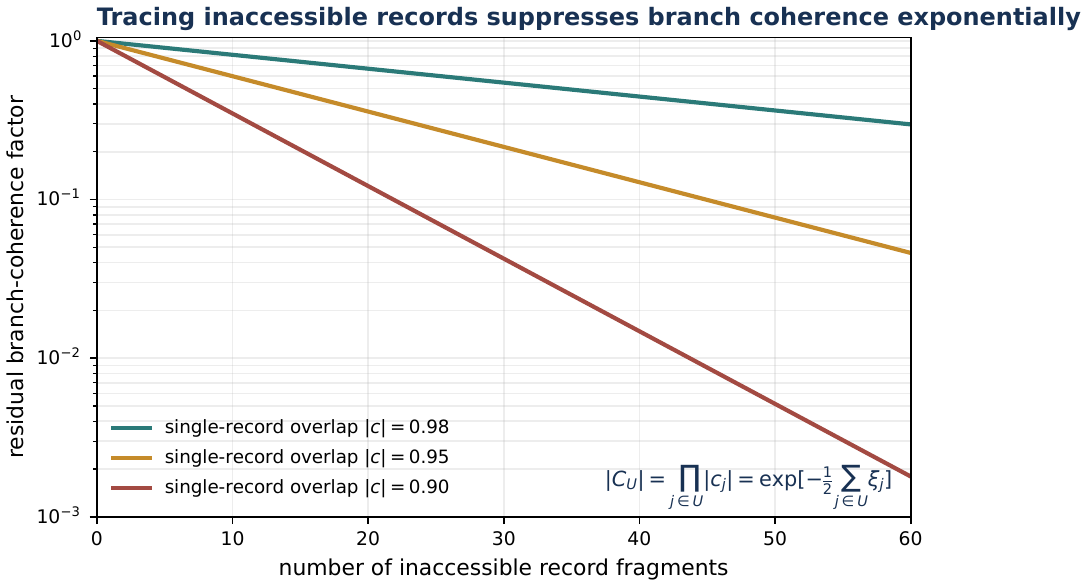}
\caption{Exact residual branch-coherence factor after tracing $|U|$ identical inaccessible records.  Even weak single-record distinguishability produces exponential suppression when many records escape control.}
\label{fig:coherence}
\end{figure}

The thermodynamic law suggests a controller-relative extension.  Let $\mathfrak R_\Gamma(t)$ be the family of partitions that an agent or apparatus $\Gamma$ can realize at or after time $t$, given locality, causal propagation, and finite control resources.  In a broadcast-like regime, define the residual factuality barrier
\begin{equation}
\mathfrak C_\Gamma(X,t)
=\inf_{\Pcal\in\mathfrak R_\Gamma(t)}\pos{\Fact_{\Pcal}^{T}(X)}.
\label{eq:causalseal}
\end{equation}
We call an event \emph{causally sealed for $\Gamma$} when $\mathfrak C_\Gamma(X,t)>0$.  This means no future admissible controller can reunite enough records to remove the positive locality penalty.  The definition is operational and compatible with global unitarity.  Whether it yields useful bounds in relativistic quantum field theory, horizon physics, or many-body systems with Lieb--Robinson cones \cite{LiebRobinson1972} is an open problem.

A second, explicitly conjectural direction is pointer selection.  In branching regimes with a fixed resource budget, one may test whether the dynamically selected pointer observable $\Pi_*$ maximizes the production rate of positive factuality,
\begin{equation}
\Pi_*\stackrel{?}{\in}\arg\max_{\Pi}
\frac{\dd}{\dd t}\Fact_{\Pcal}^{T}(X_\Pi{:}\bF_t).
\label{eq:pointerconjecture}
\end{equation}
This is not assumed in any theorem above.  It is a falsifiable proposal for connecting einselection to thermodynamic accessibility rather than to decoherence rate alone.

\section{Relation to prior work}
\label{sec:prior}

Several ingredients have established precedents.  Total correlation quantifies departure from a product state \cite{Watanabe1960}; global correlations can support work extraction that is unavailable under local operations \cite{Oppenheim2002,Groisman2005,Horodecki2005}; modular processing can dissipate correlation free energy \cite{Boyd2018}; and recent work gives general global--local control-cost relations for multipartite quantum systems \cite{Guan2025}.  Accordingly, \cref{eq:localitygap} should be read as the reversible-reset specialization that fixes the operational baseline, not as a claim that correlation free energy is newly discovered.

Quantum Darwinism and spectrum broadcast structure identify redundant, independently accessible records as the mechanism of objectivity \cite{Zurek2009,ZwolakRiedelZurek2014,Korbicz2014,Brandao2015,LeOlayaCastro2019,Korbicz2021,Touil2022}.  Latune and Elouard analyze the thermodynamic cost of measurements in a bath that redundantly records the outcome \cite{LatuneElouard2025}.  Complexity-constrained thermodynamics and dynamical min-entropy quantify erasure when the permitted process itself is restricted \cite{Munson2025,Badhani2026}.  Classical and quantum interventions are also thermodynamically constrained \cite{MilburnShrapnel2018}.

The present construction differs in four linked respects.  First, it compares access architectures on the \emph{same} record cloud.  Second, it subtracts a source-assisted baseline, isolating correlations specifically attributable to the event.  Third, the resulting signed charge has exact broadcast and secret-sharing endpoints.  Fourth, standard Darwinian redundancy immediately yields the work bound \cref{eq:landauerdarwin}.  These steps turn objectivity from a statement solely about accessible information into a quantitative work asymmetry under local control.

\section{Scope and limitations}
\label{sec:limitations}

The law is exact but not assumption free.  The principal limitations are as follows.

First, reversible free-energy differences describe asymptotic or quasistatic control.  One-shot errors, finite baths, coherence constraints, and finite-time protocols require smooth entropies or additional dissipation terms.  Second, ``partition local'' means genuinely independent block protocols.  Classical communication, a shared controller memory, or low-depth cross-block gates can partially unlock correlations and interpolate between the two baselines.  Third, the source $X$ is classical side information treated catalytically; the cost of establishing the source channel is outside the double difference and must be kept identical across comparisons.  Fourth, nonadditive interaction Hamiltonians generate energetic cross terms, so the clean entropy-only gap must be replaced by a correlation free-energy expression.  Fifth, $\Aarch_X^{\Pcal}$ is a signed architecture diagnostic, not a unique partial-information redundancy measure and not a monotone under all local channels.  Its sign measures a net balance; away from the rigid endpoints, zero does not uniquely imply nonredundant storage and a positive value alone does not replace an explicit accessibility criterion.

Most importantly, no theorem here derives the Born rule, modifies the Schr\"odinger equation, or introduces a physical collapse process.  The exact result is narrower and testable: once an event has a specified record architecture, restricted access produces a calculable reversible work asymmetry.  The causal-sealing and pointer-selection ideas are proposed extensions, not established consequences.

\section*{finite-resolution factuality certification}

The exact work identity is experimentally useful only if its sign and
magnitude are stable under state-preparation and tomography error.  A
dimension-explicit continuity certificate follows directly from sharp
entropy continuity.

\begin{theorem}[Continuity of thermodynamic factuality]
\label{thm:factuality-continuity}
Let \(\rho_{X\bF}\) and \(\sigma_{X\bF}\) be two finite-dimensional cq
record states on total Hilbert-space dimension \(D\), and suppose
\[
 \frac12\|\rho-\sigma\|_1\le\varepsilon
 \le 1-\frac1D.
\]
For a partition \(\Pcal\) into \(m\) blocks, define
\[
 g_D(\varepsilon)
 =\varepsilon\log(D-1)+h_2(\varepsilon).
\]
Then
\[
 \left|
 \Fact_{\Pcal}^{T}(X)_\rho-\Fact_{\Pcal}^{T}(X)_\sigma
 \right|
 \le
 3(m+1)\kB T\,g_D(\varepsilon).
\]
In particular, an observed positive factuality larger than the
right-hand side certifies a positive locality gap for every state in the
trace-distance error ball.
\end{theorem}

\begin{proof}
Trace distance cannot increase under partial trace.  The
Audenaert--Fannes bound therefore changes the entropy of every marginal by
at most \(g_D(\varepsilon)\), using the total dimension as a common upper
bound.  Each mutual information is a signed sum of three entropies, so its
change is at most \(3g_D(\varepsilon)\).  The architecture quantity is the
sum of \(m\) local mutual informations minus one joint mutual information.
Multiplication by \(\kB T\) proves the result.
\end{proof}

The bound is conservative because it ignores the smaller dimensions of
individual fragments.  Replacing each occurrence of \(g_D\) by the
corresponding marginal-dimension bound gives the sharper experimental
certificate.  No asymptotic or commuting-state assumption is used.

\section*{exact correction for interacting Hamiltonians}

The entropy-only locality gap has a closed extension when the Hamiltonian
contains cross-block interactions.

\begin{theorem}[Correlation free energy with interactions]
\label{thm:frontier-locality-interaction}
Let
\[
 H=\sum_{i=1}^m H_i+H_{\rm int}
\]
and let \(\rho\) be a state with marginals \(\rho_i\).  For the
nonequilibrium free energy
\(F_H(\omega)=\operatorname{Tr}(H\omega)-k_{\rm B}T S(\omega)\),
\begin{align*}
 &F_H(\rho)-F_H\!\left(\bigotimes_i\rho_i\right)\\
 &\quad=
 k_{\rm B}T\,D\!\left(
 \rho\,\middle\|\,\bigotimes_i\rho_i\right)\\
 &\qquad+
 \operatorname{Tr}\!\left[
 H_{\rm int}\!
 \left(\rho-\bigotimes_i\rho_i\right)\right].
\end{align*}
The interaction correction obeys
\[
 \left|\operatorname{Tr}\!\left[
 H_{\rm int}\!
 \left(\rho-\bigotimes_i\rho_i\right)\right]\right|
 \le
 2\|H_{\rm int}\|.
\]
\end{theorem}

\begin{proof}
The local energy terms cancel because \(\rho\) and
\(\bigotimes_i\rho_i\) have the same marginals.  The entropy difference
is the total correlation,
\[
 D\!\left(\rho\middle\|\bigotimes_i\rho_i\right)
 =\sum_iS(\rho_i)-S(\rho).
\]
This proves the identity.  H\"older's inequality and
\(\|\rho-\bigotimes_i\rho_i\|_1\le2\) give the bound.
\end{proof}

Thus the nonadditive-Hamiltonian case is not an uncontrolled caveat: the
same factuality law holds with one explicitly measurable interaction
energy, and the entropy-only sign is certified whenever its margin
exceeds \(2\|H_{\rm int}\|\).

\section*{Correlation-adaptive interaction control}

The universal \(2\|H_{\rm int}\|\) correction can be sharpened whenever
the record cloud has small or moderate total correlation.

\begin{theorem}[Pinsker-refined locality gap]
\label{thm:final-locality-pinsker}
In the setting of Theorem~\ref{thm:frontier-locality-interaction}, put
\[
 D_\rho=
 D\!\left(\rho\,\middle\|\,\bigotimes_i\rho_i\right).
\]
Then the interaction correction satisfies
\[
 \left|\operatorname{Tr}\!\left[
 H_{\rm int}\left(\rho-\bigotimes_i\rho_i\right)\right]\right|
 \le
 \|H_{\rm int}\|\min\{2,\sqrt{2D_\rho}\}.
\]
Consequently
\[
 F_H(\rho)-F_H\!\left(\bigotimes_i\rho_i\right)
 \ge
 k_{\rm B}T D_\rho-\|H_{\rm int}\|\sqrt{2D_\rho}.
\]
In particular, the interacting correlation free energy is strictly
positive whenever
\[
 D_\rho>2\left(\frac{\|H_{\rm int}\|}{k_{\rm B}T}\right)^2.
\]
\end{theorem}

\begin{proof}
H\"older's inequality bounds the correction by
\(\|H_{\rm int}\|\|\rho-\otimes_i\rho_i\|_1\).  Quantum Pinsker gives
\(\|\rho-\otimes_i\rho_i\|_1\le\sqrt{2D_\rho}\), while the trivial trace
distance bound is two.  Insert the sharper estimate into the exact
identity of Theorem~\ref{thm:frontier-locality-interaction}.  The last condition
is the positive-root inequality for the resulting quadratic in
\(\sqrt{D_\rho}\).
\end{proof}

The interaction tolerance now scales with the actually observed
correlation rather than with a state-independent worst case.

\section{Conclusion}
\label{sec:conclusion}

A redundant record is not merely the same information written many times.  It is one logical variable distributed so that no local controller can exploit all of its correlations.  This distinction gives objective evidence a precise thermodynamic signature.  The locality gap measures the total correlation hidden from a partitioned controller; subtracting the source-assisted gap removes unrelated correlations and yields the exact work and entropy-production factuality law
\begin{equation}
\Fact_{\Pcal}^{T}(X)
=\kB T\left[\sum_{B\in\Pcal}I(X{:}F_B)-I(X{:}\bF)\right].
\end{equation}
Perfect broadcast maximizes this charge, unique storage makes it vanish, and perfect secret sharing minimizes it.  Equivalently, the thermodynamic record number $\mathcal N_X^{\Pcal}=1+\Fact_{\Pcal}^{T}/(\kB T H(X))$ assigns the exact values $m$, $1$, and $0$ to those three architectures.  Quantum-Darwinism redundancy therefore implies a quantitative Landauer-scale work barrier, while no-broadcasting explains why its maximal endpoint is classical.

The central conceptual claim can be stated without metaphor: \emph{a thermodynamically public fact is a classical variable that is independently decodable from multiple record blocks and whose source loss creates a positive local--global erasure double difference}.  Whether the resulting controller-relative barrier can be promoted to a useful theory of causal irreversibility in quantum field theory remains open.  The finite-system law itself is closed, exact under explicit assumptions, numerically verified, and directly testable.

\appendix

\section{Information identities and equality conditions}
\label{app:information}

For a cq state, the mutual information is the Holevo quantity
\begin{equation}
I(X{:}Y)=S(\rho_Y)-\sum_xp_xS(\rho_{Y|x}).
\label{eq:cqmutual}
\end{equation}
Substituting \cref{eq:cqmutual} into \cref{eq:architecture} gives
\begin{align}
\Aarch_X^{\Pcal}
&=\sum_{B\in\Pcal}\left[S(F_B)-\sum_xp_xS(F_B)_{\rho_x}\right]
\notag\\
&\quad-\left[S(\bF)-\sum_xp_xS(\bF)_{\rho_x}\right]
\notag\\
&=\Tcorr_{\Pcal}(\bF)-\Tcorr_{\Pcal}(\bF|X),
\end{align}
which proves \cref{eq:decomposition}.

To prove the lower bound, use $I(X{:}F_B)\ge0$ and $I(X{:}\bF)\le H(X)$:
\begin{equation}
\Aarch_X^{\Pcal}\ge-H(X).
\end{equation}
For the upper bound, choose a block $B_*$ maximizing $I(X{:}F_B)$.  Data processing gives
\begin{equation}
I(X{:}\bF)\ge I(X{:}F_{B_*}),
\end{equation}
while every remaining block obeys $I(X{:}F_B)\le H(X)$.  Therefore
\begin{align}
\Aarch_X^{\Pcal}
&=\sum_{B\ne B_*}I(X{:}F_B)
+I(X{:}F_{B_*})-I(X{:}\bF)\notag\\
&\le(m-1)H(X).
\end{align}

Lower saturation requires simultaneously $I(X{:}F_B)=0$ for every block and $I(X{:}\bF)=H(X)$.  For cq states, the former is equivalent to $\rho_{F_B|x}=\rho_{F_B}$ for every $x$ with $p_x>0$; the latter is equivalent to perfect distinguishability of the global conditional states.  Upper saturation requires $I(X{:}F_B)=H(X)$ for every $B$.  Because each block is a channel output of the full cloud, this condition is also sufficient.  For a finite cq ensemble with nonzero priors, $I(X{:}F_B)=H(X)$ if and only if the supports of $\rho_{F_B|x}$ and $\rho_{F_B|x'}$ are orthogonal for $x\ne x'$.

For any anchor $B_0$, direct rearrangement gives
\begin{align}
(m-1)H(X)-\Aarch_X^{\Pcal}
&=\sum_{B\ne B_0}[H(X)-I(X{:}F_B)]\notag\\
&\quad+I(X{:}\bF)-I(X{:}F_{B_0}),\\
\Aarch_X^{\Pcal}+H(X)
&=\sum_{B\in\Pcal}I(X{:}F_B)+H(X)-I(X{:}\bF).
\end{align}
Data processing makes all terms nonnegative.  Multiplication by $\kB T$ proves \cref{eq:upperrigidity,eq:lowerrigidity} and their stability consequences.

\subsection{Invariance and additivity}

Local isometries preserve every entropy appearing in \cref{eq:architecture}; hence they preserve $\Aarch_X^{\Pcal}$ and $\Fact_{\Pcal}^{T}$.  Consider two statistically independent record systems, $(X,\bF)$ and $(Y,\bG)$, with product state and the disjoint-union partition $\Pcal\sqcup\mathcal Q$.  Mutual information is additive under tensor products, giving
\begin{equation}
\Aarch_{XY}^{\Pcal\sqcup\mathcal Q}(\bF\bG)
=\Aarch_X^{\Pcal}(\bF)+\Aarch_Y^{\mathcal Q}(\bG).
\end{equation}
The same relation holds for thermodynamic factuality at a common bath temperature.

\section{Proof of the thermodynamic laws}
\label{app:thermo}

With additive Hamiltonian and product blank,
\begin{equation}
\Ffree_T(\omega_{\bF})=\sum_{B\in\Pcal}\Ffree_T(\omega_{F_B}),
\end{equation}
and
\begin{equation}
\Tr(H_{\bF}\rho_{\bF})=\sum_{B\in\Pcal}\Tr(H_{F_B}\rho_{F_B}).
\end{equation}
Subtracting \cref{eq:wglobal} from \cref{eq:wlocal}, all blank and energy terms cancel:
\begin{align}
G_{\Pcal}^{\emptyset}
&=\Ffree_T(\rho_{\bF})-
\sum_{B\in\Pcal}\Ffree_T(\rho_{F_B})\notag\\
&=\kB T\left[\sum_{B\in\Pcal}S(F_B)-S(\bF)\right]\notag\\
&=\kB T\Tcorr_{\Pcal}(\bF).
\end{align}
This proves \cref{eq:localitygap}.

When $X$ is supplied, the same calculation is performed for every conditional state and averaged over $x$:
\begin{align}
G_{\Pcal}^{X}
&=\kB T\sum_xp_x
\left[\sum_{B\in\Pcal}S(F_B)_{\rho_x}-S(\bF)_{\rho_x}\right]\notag\\
&=\kB T\Tcorr_{\Pcal}(\bF|X).
\end{align}
Using \cref{eq:decomposition},
\begin{equation}
G_{\Pcal}^{\emptyset}-G_{\Pcal}^{X}
=\kB T\Aarch_X^{\Pcal}(\bF),
\end{equation}
which proves \cref{eq:factuality}; multiplying \cref{eq:architecturebounds} by $\kB T$ proves \cref{eq:bounds}.  Since entropy production for an isothermal transformation is $(W-\Delta \Ffree_T)/T$ and the global baselines attain the full-cloud free-energy change reversibly, \cref{eq:entropyproduction} follows immediately.

The proof also clarifies the role of energetic assumptions.  If the physical Hamiltonian contains cross-block interactions, the local--global difference becomes
\begin{equation}
G_{\Pcal}^{\emptyset}
=\kB T\Tcorr_{\Pcal}(\bF)+\Delta E_{\mathrm{int}}+\Delta F_{\mathrm{blank,int}},
\end{equation}
where the final two terms depend on how interactions are switched and on the blank Hamiltonian.  The entropy-only law is recovered when those terms vanish or are separately calibrated.

\section{Derivation of the objectivity bounds}
\label{app:bounds}

From \cref{eq:redundancycondition},
\begin{equation}
\sum_{j=1}^{R}I(X{:}F_{B_j})
\ge R(1-\delta)H(X).
\end{equation}
Since the union cannot contain more than $H(X)$ nats about a classical source,
\begin{equation}
\Aarch_X^{\Pcal}
\ge\left[R(1-\delta)-1\right]H(X),
\end{equation}
which proves \cref{eq:landauerdarwin}.  Also $\Tcorr=\Aarch+\Tcorr(\cdot|X)\ge\Aarch$ and $\Tcorr\ge0$, proving \cref{eq:rawbound}.

For a decoder $\widehat X_B$ obtained by measuring block $B$, data processing gives
\begin{equation}
I(X{:}F_B)\ge I(X{:}\widehat X_B).
\end{equation}
Fano's inequality gives
\begin{equation}
H(X|\widehat X_B)
\le h_2(p_{e,B})+p_{e,B}\ln(d_X-1),
\end{equation}
and hence \cref{eq:fano} after summing over blocks and using $I(X{:}\bF)\le H(X)$.

For the spectrum broadcast state \cref{eq:sbs}, orthogonality implies
\begin{align}
S(F_j)&=H(X)+\sum_xp_xS(\rho_{F_j|x}),\\
S(\bF)&=H(X)+\sum_xp_x\sum_{j=1}^{R}S(\rho_{F_j|x}).
\end{align}
Therefore $\Tcorr(\bF)=(R-1)H(X)$.  Conditional product structure gives $\Tcorr(\bF|X)=0$, proving \cref{eq:sbssaturation}.

\section{Chain rules and partition refinement}
\label{app:chainrules}

The total-correlation increment follows directly:
\begin{align}
\Tcorr(F_{1:n+1})-\Tcorr(F_{1:n})
&=S(F_{n+1})+S(F_{1:n})-S(F_{1:n+1})\notag\\
&=I(F_{n+1}{:}F_{1:n}).
\end{align}
For factuality,
\begin{align}
\Aarch_X(F_{1:n+1})-\Aarch_X(F_{1:n})
&=I(X{:}F_{n+1})\notag\\
&\quad-I(X{:}F_{n+1}|F_{1:n}),
\end{align}
which is \cref{eq:factincrementa}.  The standard interaction-information identity
\begin{equation}
I(X{:}A)-I(X{:}A|B)=I(A{:}B)-I(A{:}B|X)
\end{equation}
then gives \cref{eq:factincrementb}.

When a fine partition contains separate blocks $B$ and $C$ and the coarse partition replaces them by $BC$, the terms unrelated to $B,C$ cancel:
\begin{align}
\Tcorr_{\mathrm{fine}}-\Tcorr_{\mathrm{coarse}}
&=S(B)+S(C)-S(BC)\notag\\
&=I(B{:}C).
\end{align}
Applying the same identity conditionally on $X$ gives \cref{eq:partitioncond}, and subtracting the two gives \cref{eq:partitionfact}.

\section{Closed-form model calculations}
\label{app:models}

\subsection{Binary-symmetric record cloud}

For $R$ conditionally independent binary-symmetric records, let $K=\sum_{j=1}^{R}F_j$.  With uniform $X$,
\begin{align}
P(K=k)
&=\frac12\binom{R}{k}
\Bigl[\epsilon^k(1-\epsilon)^{R-k}\notag\\
&\qquad +(1-\epsilon)^k\epsilon^{R-k}\Bigr],\\
q_k\equiv P(X=1|K=k)
&=\left[1+
\left(\frac{\epsilon}{1-\epsilon}\right)^{2k-R}
\right]^{-1}.
\end{align}
Therefore
\begin{equation}
H(X|F_{1:R})
=\sum_{k=0}^{R}P(K=k)h_2(q_k),
\label{eq:bscsum}
\end{equation}
which inserted into \cref{eq:bsc} yields the plotted curves.

\subsection{Pure-state collision cloud}

For two equal-prior pure states $|\psi_0\rangle,|\psi_1\rangle$ with overlap $q$, the nonzero eigenvalues of
\begin{equation}
\frac12\left(|\psi_0\rangle\!\langle\psi_0|
+|\psi_1\rangle\!\langle\psi_1|\right)
\end{equation}
are $(1\pm q)/2$.  A single record therefore has entropy $h_2[(1+c)/2]$, whereas the joint conditional states have overlap $c^R$ and the cloud entropy is $h_2[(1+c^R)/2]$.  Substitution into \cref{eq:totalcorr} proves \cref{eq:collision}.

\subsection{Parity cloud}

Unconditionally, the $R$ fragments are independent fair bits, so $S(F_j)=\ln2$, $S(\bF)=R\ln2$, and $\Tcorr(\bF)=0$.  Conditioning on the parity fixes one global constraint: $S(\bF|X)=(R-1)\ln2$, while each conditional marginal remains fair.  Thus $\Tcorr(\bF|X)=\ln2$ and \cref{eq:decomposition} gives $\Aarch=-\ln2$.

\section{Numerical methods and reproducibility}
\label{app:numerics}

Classical trials sampled a random source alphabet of size $2$--$4$, two to five fragments of dimension $2$--$4$, and a strictly positive joint distribution from independent gamma variates.  For every trial the code evaluated $\Tcorr$, $\Tcorr(\cdot|X)$, all local mutual informations, and the joint mutual information directly from marginal distributions.  Independent random local energies, temperatures, and product blanks were then used to compare the free-energy work difference with $\kB T\Tcorr$.

Quantum trials sampled two or three source values and random density matrices on two to four qubits from normalized complex Wishart matrices of random rank.  Entropies were computed from Hermitian eigenspectra after numerical positivity clipping below $10^{-12}$.  Partial traces were performed tensorially, and no commutativity was imposed on the conditional states.

For seed $20260728$, the largest classical and quantum decomposition residuals were $2.67\times10^{-15}$ and $7.78\times10^{-16}$; the largest raw and source-loss work-gap residuals were $1.34\times10^{-15}$ and $1.00\times10^{-14}$; the endpoint-rigidity residuals were at most $1.78\times10^{-15}$; and the perfect-broadcast, parity-sharing, and collision-model residuals were $0$, $3.34\times10^{-16}$, and $1.78\times10^{-15}$, respectively.  Every tested inequality had nonnegative margin.  The reproducibility archive contains the exact scripts, package-light source, generated figures, and JSON output.  Numerical tests are not used as evidence for an unproved equality; they guard against implementation errors and verify the noncommuting examples.

\section{Residual coherence after inaccessible records}
\label{app:coherence}

Partition the environmental fragments in \cref{eq:branchstate} into accessible $A$ and inaccessible $U$.  Writing
\begin{equation}
|E_x\rangle=|E_x^A\rangle|E_x^U\rangle,
\end{equation}
the reduced state after tracing $U$ contains the off-diagonal term
\begin{equation}
\frac12|0\rangle\!\langle1|
\otimes|E_0^A\rangle\!\langle E_1^A|
\langle E_1^U|E_0^U\rangle.
\end{equation}
Product structure gives
\begin{equation}
\langle E_1^U|E_0^U\rangle
=\prod_{j\in U}\langle e_j^1|e_j^0\rangle,
\end{equation}
and taking the modulus proves \cref{eq:coherencefactor}.  No thermodynamic assumption enters this calculation.

\renewcommand{\bibfont}{\footnotesize}
\setlength{\bibsep}{0pt plus 0.3ex}
\bibliography{references}

\end{document}